\documentclass[11pt]{article}

\usepackage{amsmath,amsthm,amsfonts,amssymb,comment}

\usepackage{booktabs,ctable,tabularx,array}
\newcolumntype{+}{>{\global\let\currentrowstyle\relax}}
\newcolumntype{^}{>{\currentrowstyle}}

\newcommand{\otoprule}{\midrule[\heavyrulewidth]}

\allowdisplaybreaks

\DeclareMathOperator{\Hom}{Hom}
\DeclareMathOperator{\Com}{Com}

\DeclareMathOperator{\1}{id}
\newcommand{\NNN}{\mathbb{N}}
\newcommand{\RR}{\mathbb{R}}

\newcommand{\EEnd}{\mathcal End}
\newcommand{\EE}{\mathcal E}
\newcommand{\De}{\Delta}

\renewcommand{\=}{:=}
\renewcommand{\t}{\otimes}
\renewcommand{\:}{\colon}

\newcommand{\m}{\overset{\circ}{\mu}}
\newcommand{\A}{\hat{A}}
\newcommand{\pp}{\hat{p}}
\newcommand{\ppp}{p_0}
\newcommand{\q}{\hat{q}}
\newcommand{\muu}{\hat{\mu}}
\newcommand{\qxi}{\hat{\xi}}
\newcommand{\Q}{\hat{A}_-}
\newcommand{\PP}{\hat{A}_+}

\newtheorem{thm}{Theorem}[section]

\newtheorem{lemma}[thm]{Lemma}

\theoremstyle{definition}
\newtheorem{defn}[thm]{Definition}

\newtheorem{rem}[thm]{Remark}

\begin{document}

\title{\bf Jacobi operators of quantum counterparts\\ of three-dimensional  
real Lie algebras \\ over the harmonic oscillator}

\author{\Large Eugen Paal and J\"{u}ri Virkepu}
\date{}

\maketitle

\begin{abstract}
Operadic Lax representations for the harmonic oscillator are used to
construct the quantum counterparts of three-dimensional real 
Lie algebras.  The Jacobi operators of these quantum algebras are
explicitly calculated.
\end{abstract}

\section{Introduction and outline of the paper}

In Hamiltonian formalism, a mechanical system is described by the
canonical variables $q^i,p_i$ and their time evolution is prescribed
by the Hamiltonian equations
\begin{equation}
\label{ham} 
\dfrac{dq^i}{dt}=\dfrac{\partial H}{\partial p_i}, \quad
\dfrac{dp_i}{dt}=-\dfrac{\partial H}{\partial q^i}
\end{equation}
By a Lax representation \cite{Lax68} of a mechanical system
one means such a pair $(L,M)$ of matrices (linear operators) $L,M$
that the above Hamiltonian system may be represented as the Lax
equation
\begin{equation}
\label{lax} 
\dfrac{dL}{dt}= ML-LM
\end{equation}
Thus, from the algebraic point of view, mechanical systems may be
represented by linear operators, i.e by  linear maps $V\to V$  of a
vector space $V$. In particular, representation of the physical observables by linear operators is used in quantum mechanics by and the Heisenberg equations. As a generalization of this one can pose the
following question \cite{Paal07}: how can the time evolution of the
linear operations (multiplications) $V^{\t n}\to V$ be described?

The algebraic operations (multiplications) can be seen as an example
of the \emph{operadic} variables \cite{Ger}. If an operadic system
depends on time one can speak about \emph{operadic dynamics}
\cite{Paal07}. The latter may be introduced by simple and natural
analogy with the Hamiltonian dynamics. In particular, the time
evolution of the operadic variables may be given by the operadic Lax
equation. In \cite{PV07,PV08,PV08-1}, the low-dimensional binary operadic Lax representations for the harmonic oscillator were constructed.
In \cite{PV08-2} it was shown how the operadic Lax representations are related to the conservation of energy.

In \cite{PV09-1}, the operadic Lax representations were used to
constuct the quantum counterparts of the real three-dimensional Lie
algebras in Bianchi classification over the harmonic oscillator. In
this paper,  the Jacobi operators of some of these quantum algebras
are explicitly calculated.

\section{Endomorphism operad and Gerstenhaber brackets}

Let $K$ be a unital associative commutative ring, $V$ be a unital
$K$-module, and $\EE_V^n\= {\EEnd}_V^n\= \Hom(V^{\t n},V)$
($n\in\NNN$). For an \emph{operation} $f\in\EE^n_V$, we refer to $n$
as the \emph{degree} of $f$ and often write (when it does not cause
confusion) $f$ instead of $\deg f$. For example, $(-1)^f\= (-1)^n$,
$\EE^f_V\=\EE^n_V$ and $\circ_f\= \circ_n$. Also, it is convenient
to use the \emph{reduced} degree $|f|\= n-1$. Throughout this paper,
we assume that $\t\= \t_K$.

\begin{defn}[endomorphism operad \cite{Ger}]
\label{HG} For $f\t g\in\EE_V^f\t\EE_V^g$ define the \emph{partial
compositions}
\[
f\circ_i g\= (-1)^{i|g|}f\circ(\1_V^{\t i}\t g\t\1_V^{\t(|f|-i)})
\quad \in\EE^{f+|g|}_V,
         \quad 0\leq i\leq |f|
\]
The sequence $\EE_V\= \{\EE_V^n\}_{n\in\NNN}$, equipped with the
partial compositions $\circ_i$, is called the \emph{endomorphism
operad} of $V$.
\end{defn}

\begin{defn}[total composition \cite{Ger}]
The \emph{total composition}
$\circ \:\EE^f_V\t\EE^g_V\to\EE^{f+|g|}_V$ is defined by
\[
f\circ g\= \sum_{i=0}^{|f|}f\circ_i g\quad \in \EE_V^{f+|g|}, \quad |\circ|=0
\]
The pair $\Com\EE_V\= \{\EE_V,\circ\}$ is called the \emph{composition
algebra} of $\EE_V$.
\end{defn}

\begin{defn}[Gerstenhaber brackets \cite{Ger}]
The  \emph{Gerstenhaber brackets} $[\cdot,\cdot]$ are defined in
$\Com\EE_V$ as a graded commutator by
\[
[f,g]\= f\circ g-(-1)^{|f||g|}g\circ f=-(-1)^{|f||g|}[g,f],\quad
|[\cdot,\cdot]|=0
\]
\end{defn}

The \emph{commutator algebra} of $\Com \EE_V$ is denoted as
$\Com^{-}\!\EE_V\= \{\EE_V,[\cdot,\cdot]\}$. One can prove (e.g
\cite{Ger}) that $\Com^-\!\EE_V$ is a \emph{graded Lie algebra}. The
Jacobi identity reads
\[
(-1)^{|f||h|}[f,[g,h]]+(-1)^{|g||f|}[g,[h,f]]+(-1)^{|h||g|}[h,[f,g]]=0
\]

\section{Operadic Lax pair}

Assume that $K\= \RR$ or $K\= \mathbb{C}$ and operations are
differentiable. Dynamics in operadic systems (operadic dynamics) may
be introduced by

\begin{defn}[operadic Lax pair \cite{Paal07}]
Allow a classical dynamical system to be described by the
Hamiltonian system \eqref{ham}. An \emph{operadic Lax pair} is a
pair $(\mu,M)$ of homogeneous operations $\mu,M\in\EE_V$, such that the
Hamiltonian system  (\ref{ham}) may be represented as the
\emph{operadic Lax equation}
\[
\frac{d\mu}{dt}=[M,\mu]\= M\circ\mu-(-1)^{|M||\mu|}\mu\circ M
\]
The pair $(L,M)$ is also called an \emph{operadic Lax representations} of/for Hamiltonian system \eqref{ham}.
\end{defn}

\begin{rem}
Evidently, the degree constraints $|M|=|L|=0$ give rise to ordinary
Lax equation (\ref{lax}) \cite{Lax68}. In this paper we assume that $|M|=0$.
\end{rem}

The Hamiltonian of the harmonic oscillator (HO) is
\[
H(q,p)=\frac{1}{2}(p^2+\omega^2 q^2)
\]
Thus, the Hamiltonian system of HO reads
\begin{equation}
\label{eq:h-osc} \frac{dq}{dt}=\frac{\partial H}{\partial p}=p,\quad
\frac{dp}{dt}=-\frac{\partial H}{\partial q}=-\omega^2q
\end{equation}
If $\mu$ is a linear algebraic operation we can use the above
Hamilton equations to obtain
\[
\dfrac{d\mu}{dt} =\dfrac{\partial\mu}{\partial
q}\dfrac{dq}{dt}+\dfrac{\partial\mu}{\partial p}\dfrac{dp}{dt}
=p\dfrac{\partial\mu}{\partial
q}-\omega^2q\dfrac{\partial\mu}{\partial p}
 =[M,\mu]
\]
Therefore, we get the following linear partial differential equation
for $\mu(q,p)$:
\begin{equation}
\label{eq:diff}
p\dfrac{\partial\mu}{\partial
q}-\omega^2q\dfrac{\partial\mu}{\partial p}=[M,\mu]
\end{equation}
By integrating \eqref{eq:diff} one can get collections of operations called \cite{Paal07} the \emph{operadic} (Lax representations for/of) harmonic oscillator.

\section{3D binary anti-commutative operadic Lax representations for harmonic oscillator}

\begin{lemma}
\label{lemma:harmonic3} 
Matrices
\[
L\=\begin{pmatrix}
    p & \omega q & 0 \\
    \omega q & -p & 0 \\
    0 & 0 & 1 \\
  \end{pmatrix},\quad
M\=\frac{\omega}{2}
\begin{pmatrix}
    0 & -1 &0\\
1 & 0 & 0\\
0 & 0 & 0
  \end{pmatrix}
\]
represent a 3D Lax representation for the harmonic
oscillator.
\end{lemma}

\begin{defn}[quasi-canonical coordinates]
For the harmonic oscillator define its  \emph{quasi-canonical coordinates} $A_\pm$ by
\begin{equation}
\label{eq:def_A}
A^{2}_+ - A^{2}_- =2p,\quad
A_+ A_-=\omega q
\end{equation}
\end{defn}

\begin{thm}[see \cite{PV08-1}]
\label{thm:main}
Let $C_{\nu}\in\mathbb{R}$ ($\nu=1,\ldots,9$) be
arbitrary real--valued parameters, such that
\begin{equation}
\label{eq:cond} C_2^2+C_3^2+C_5^2+C_6^2+C_7^2+C_8^2\neq0
\end{equation}
Let $M$ be defined as in Lemma \ref{lemma:harmonic3} and $\mu:
V\otimes V\to V$ be an anti-commutative binary operation in a 3D real
vector space $V$ with the structure constants (functions)
\begin{equation}
\label{eq:theorem}
\begin{cases}
\mu_{11}^{1}=\mu_{22}^{1}=\mu_{33}^{1}=\mu_{11}^{2}=\mu_{22}^{2}=\mu_{33}^{2}=\mu_{11}^{3}=\mu_{22}^{3}=\mu_{33}^{3}=0\\
\mu_{23}^{1}=-\mu_{32}^{1}=C_2p-C_3\omega q-C_4\\
\mu_{13}^{2}=-\mu_{31}^{2}=C_2p-C_3\omega q+C_4\\
\mu_{31}^{1}=-\mu_{13}^{1}=C_2\omega q+C_3p-C_1\\
\mu_{23}^{2}=-\mu_{32}^{2}=C_2\omega q+C_3p+C_1\\
\mu_{12}^{1}=-\mu_{21}^{1}=C_5 A_+ + C_6 A_-\\
\mu_{12}^{2}=-\mu_{21}^{2}=C_5 A_- - C_6 A_+\\
\mu_{13}^{3}=-\mu_{31}^{3}=C_7 A_+ + C_8 A_-\\
\mu_{23}^{3}=-\mu_{32}^{3}=C_7 A_- - C_8 A_+\\
\mu_{12}^{3}=-\mu_{21}^{3}=C_9
\end{cases}
\end{equation}
Then $(\mu,M)$ is an operadic Lax pair for the harmonic oscillator.
\end{thm}

\section{Initial conditions}

Now specify the coefficients $C_{\nu}$ in Theorem \ref{thm:main} by the
initial conditions
\[
\left. \mu\right|_{t=0}=\m{}_,\quad
\left.p\right|_{t=0}
=p_0,\quad \left. q\right|_{t=0}=0
\]
Denoting $E\=H|_{t=0}$, the latter together with \eqref{eq:def_A} yield the initial
conditions for $A_\pm$:
\[
\begin{cases}
\left.\left(A^{2}_+ + A^{2}_-\right)\right|_{t=0}=2\sqrt{2E}\\
\left.\left(A^{2}_+ - A^{2}_- \right)\right|_{t=0}=2p_0\\
\left.A_+ A_-\right|_{t=0}=0
\end{cases}
\quad \Longleftrightarrow \quad
\begin{cases}
\ppp\!\!>0\\
\left.A^{2}_+\right|_{t=0}=2p_0\\
\left.A_- \right|_{t=0}=0
\end{cases}
\vee\quad
\begin{cases}
\ppp<0\\
\left.A_+\right|_{t=0}=0\\
\left.A_-^2\right|_{t=0}=-2p_0
\end{cases}
\]
In what follows assume that $p_0>0$ and $A_+|_{t=0}=\sqrt{2p_0}$. The other cases
can be treated similarly. Note that in this case $p_0=\sqrt{2E}$. From \eqref{eq:theorem} we get the following linear system:
\begin{equation}
\label{eq:constants}
\left\{
  \begin{array}{lll}
C_1=\frac{1}{2}\left(\overset{\circ}{\mu}{}_{23}^{2}-\overset{\circ}{\mu}{}_{31}^{1}\right),&
C_2=\frac{1}{2\ppp}\left(\overset{\circ}{\mu}{}_{13}^{2}+\overset{\circ}{\mu}{}_{23}^{1}\right),&
C_3=\frac{1}{2\ppp}\left(\overset{\circ}{\mu}{}_{23}^{2}+\overset{\circ}{\mu}{}_{31}^{1}\right)\vspace{1mm}\\
C_4=\frac{1}{2}\left(\overset{\circ}{\mu}{}_{13}^{2}-\overset{\circ}{\mu}{}_{23}^{1}\right),&
C_5=\frac{1}{\sqrt{2\ppp}}\overset{\circ}{\mu}{}_{12}^{1},&
C_6=-\frac{1}{\sqrt{2\ppp}}\overset{\circ}{\mu}{}_{12}^{2}\vspace{1mm}\\
C_7=\frac{1}{\sqrt{2\ppp}}\overset{\circ}{\mu}{}_{13}^{3},&
C_8=-\frac{1}{\sqrt{2\ppp}}\overset{\circ}{\mu}{}_{23}^{3},&
C_9=\overset{\circ}{\mu}{}_{12}^{3}
\end{array}
\right.
\end{equation}

\section{Bianchi classification of 3D real Lie algebras}
We use the Bianchi classification of 3D real Lie
algebras \cite{Landau80}. The structure equations of the latter can
be presented as follows:
\[
[e_1,e_2]=-\alpha e_2+n^{3}e_3,\quad
[e_2,e_3]=n^{1}e_1,\quad
[e_3,e_1]=n^{2}e_2+\alpha e_3
\]
The values of the parameters $\alpha,n^{1}, n^{2},n^{3}$ and the
corresponding structure constants are presented in Table
\ref{table:Bianchi1}.

\begin{table}[htpb]
\newcolumntype{Y}{>{\centering\arraybackslash}X}
\renewcommand{\tabularxcolumn}[1]{>{\arraybackslash}m{#1}}
\begin{tabularx}{\textwidth}{{>{\hsize=1.65\hsize}Y>{\hsize=0.30\hsize}Y>{\hsize=2.4\hsize}Y>{\hsize=0.85\hsize}Y>{\hsize=0.85\hsize}Y
>{\hsize=0.85\hsize}Y>{\hsize=0.85\hsize}Y>{\hsize=0.85\hsize}Y>{\hsize=0.85\hsize}Y>{\hsize=0.85\hsize}Y>{\hsize=0.85\hsize}Y>{\hsize=0.85\hsize}Y}}
\toprule%
Bianchi type & $\alpha$ & $(n^{1},n^{2},n^{3})$ &
$\overset{\circ}{\mu}{}_{12}^{1}$ &
$\overset{\circ}{\mu}{}_{12}^{2}$ &
$\overset{\circ}{\mu}{}_{12}^{3}$ &
 $\overset{\circ}{\mu}{}_{23}^{1}$ & $\overset{\circ}{\mu}{}_{23}^{2}$ & $\overset{\circ}{\mu}{}_{23}^{3}$
  & $\overset{\circ}{\mu}{}_{31}^{1}$ & $\overset{\circ}{\mu}{}_{31}^{2}$ &
  $\overset{\circ}{\mu}{}_{31}^{3}$\\\otoprule
 I & 0 & $(0,0,0)$ & 0 & 0 & 0 & 0 & 0 & 0 & 0 & 0 & 0
\\\midrule
II & 0 & $(1,0,0)$ & 0 & 0 & 0 & $1$ & 0 & 0 & 0 & 0 & 0
\\\midrule
VII & 0 & $(1,1,0)$ & 0 & 0 & 0 & $1$ & 0 & 0 & 0 & $1$ & 0
\\\midrule
VI & 0 & $(1,-1,0)$ & 0 & 0 & 0 & $1$ & 0 & 0 & 0 & $-1$ & 0
\\\midrule
IX & 0 & $(1,1,1)$ & 0 & 0 & $1$ & $1$ & 0 & 0 & 0 & $1$ & 0
\\\midrule
VIII & 0 & $(1,1,-1)$ & 0 & 0 & $-1$ & $1$ & 0 & 0 & 0 & $1$ & 0
\\\midrule
V & 1 & $(0,0,0)$ & 0 & $-1$ & 0 & 0 & 0 & 0 & 0 & 0 & $1$
\\\midrule
IV & 1 & $(0,0,1)$ & 0 & $-1$ & $1$ & 0 & 0 & 0 & 0 & 0 & $1$
\\\midrule
VII$_{a}$ & $a$ & $(0,1,1)$ & 0 & $-a$ & $1$ & 0 & 0 & 0 & 0 & $1$ &
$a$
\\\midrule
III$_{a=1}$& 1 & $(0,1,-1)$ & 0 & $-1$ & $-1$ & 0 & 0 & 0 & 0 & $1$
& $1$
\\\midrule
VI$_{a\neq 1}$& $a$ & $(0,1,-1)$ & 0 & $-a$ & $-1$ & 0 & 0 & 0 & 0 &
$1$ & $a$
\\\bottomrule
\end{tabularx}
\caption{3D real Lie algebras in Bianchi classification. Here $a>0$}
\label{table:Bianchi1}
\end{table}

\section{Dynamical deformations of 3D real Lie algebras}

By using the structure constants of the 3D real Lie
algebras in the Bianchi classification, Theorem \ref{thm:main} and
relations \eqref{eq:constants} one  can propose that the time evolution of the
3D real Lie algebras is prescribed \cite{PV08-2} as
given in Table \ref{table:Bianchi3}.
\begin{table}[htpb]
\newcolumntype{Y}{>{\centering\arraybackslash}X}
\renewcommand{\tabularxcolumn}[1]{>{\arraybackslash}m{#1}}
\begin{tabularx}{\textwidth}{{>{\hsize=2\hsize}Y>{\hsize=0.75\hsize}Y>{\hsize=0.9\hsize}Y>{\hsize=0.5\hsize}Y>{\hsize=0.85\hsize}Y
>{\hsize=0.8\hsize}Y>{\hsize=0.85\hsize}Y>{\hsize=0.90\hsize}Y>{\hsize=0.8\hsize}Y>{\hsize=0.7\hsize}Y>{\hsize=0.95\hsize}Y}}
\toprule%
Dynamical Bianchi type & $\mu_{12}^{1}$ & $\mu_{12}^{2}$ &
$\mu_{12}^{3}$ & $\mu_{23}^{1}$ & $\mu_{23}^{2}$ & $\mu_{23}^{3}$ &
$\mu_{31}^{1}$ & $\mu_{31}^{2}$ &  $\mu_{31}^{3}$
\\\otoprule
I$^{t}$ & 0 & 0 & 0 & 0 & 0 & 0 & 0 & 0 & 0
\\  \midrule
II$^{t}$ & 0 & 0 & 0 & $\frac{p+p_0}{2p_0}$ & $\frac{\omega
q}{2p_0}$ & 0 & $\frac{\omega q}{2p_0}$ & $\frac{p-p_0}{-2p_0}$ & 0
\\ \midrule
VII$^{t}$ & 0 & 0 & 0 & $1$ & 0 & 0 & 0 & $1$ & 0
\\ \midrule
VI$^{t}$ & 0 & 0 & 0 & $\frac{p}{p_0}$& $\frac{\omega q}{p_0}$ & 0 &
$\frac{\omega q}{p_0}$ & $-\frac{p}{p_0}$ & 0
\\\midrule
IX$^{t}$ & 0 & 0 & $1$ & $1$ & 0 & 0 & 0 & $1$ & 0
\\  \midrule
VIII$^{t}$ & 0 & 0 & $-1$ & $1$ & 0 & 0 & 0 & $1$ & 0
\\\midrule
V$^{t}$ & $\frac{A_-}{\sqrt{2p_0}}$ & $\frac{-A_+}{\sqrt{2p_0}}$ & 0
& 0 & 0 & $\frac{-A_-}{\sqrt{2p_0}}$ & 0 & 0 &
$\frac{A_+}{\sqrt{2p_0}}$
\\\midrule
IV$^{t}$ & $\frac{A_-}{\sqrt{2p_0}}$ & $\frac{-A_+}{\sqrt{2p_0}}$ &
$1$ & 0 & 0 & $\frac{-A_-}{\sqrt{2p_0}}$ & 0 & 0 &
$\frac{A_+}{\sqrt{2p_0}}$
\\\midrule
VII$_{a}^{t}$ & $\frac{aA_-}{\sqrt{2p_0}}$ &
$\frac{-aA_+}{\sqrt{2p_0}}$ & $1$ & $\frac{p-p_0}{-2p_0}$ &
$\frac{\omega q}{-2p_0}$ & $\frac{-aA_-}{\sqrt{2p_0}}$ &
$\frac{\omega q}{-2p_0}$ & $\frac{p+p_0}{2p_0}$ &
$\frac{aA_+}{\sqrt{2p_0}}$
\\ \midrule
III$_{a=1}^{t}$ & $\frac{A_-}{\sqrt{2p_0}}$ &
$\frac{-A_+}{\sqrt{2p_0}}$ & $-1$ & $\frac{p-p_0}{-2p_0}$ &
$\frac{\omega q}{-2p_0}$ & $\frac{-A_-}{\sqrt{2p_0}}$ &
$\frac{\omega q}{-2p_0}$ & $\frac{p+p_0}{2p_0}$ &
$\frac{A_+}{\sqrt{2p_0}}$
\\ \midrule
VI$_{a\neq1}^{t}$ & $\frac{aA_-}{\sqrt{2p_0}}$ &
$\frac{-aA_+}{\sqrt{2p_0}}$ & $-1$ & $\frac{p-p_0}{-2p_0}$ &
$\frac{\omega q}{-2p_0}$ & $\frac{-aA_-}{\sqrt{2p_0}}$ &
$\frac{\omega q}{-2p_0}$ & $\frac{p+p_0}{2p_0}$ &
$\frac{aA_+}{\sqrt{2p_0}}$
\\  \bottomrule
\end{tabularx}
\caption{Time evolution of 3D real Lie algebras. 
Here $p_0=\sqrt{2E}$}
\label{table:Bianchi3}
\end{table}

\section{Quantum counterparts of 3D real Lie algebras}

Let now the harmonic oscillator be quantized, i.e its canonical
coordinates satisfy the CCR
\[
[\q,\q]=0=[\pp,\pp],\quad [\pp,\q]=\frac{\hbar}{i}
\]
Then the classical observables $A_\pm(q,p)$ will be quantized as
well and  their quantum counterparts are denoted by
$\hat{A}_\pm:=A_\pm(\q,\pp)$. As the result, the quantum
counterparts of the 3D real Lie algebras can be
listed as presented in Table \ref{table:Bianchi4}.

\begin{table}[htpb]
\newcolumntype{Y}{>{\centering\arraybackslash}X}
\renewcommand{\tabularxcolumn}[1]{>{\arraybackslash}m{#1}}
\begin{tabularx}{\textwidth}{{>{\hsize=2\hsize}Y>{\hsize=0.75\hsize}Y>{\hsize=0.9\hsize}Y>{\hsize=0.65\hsize}Y>{\hsize=0.85\hsize}Y
>{\hsize=0.8\hsize}Y>{\hsize=0.9\hsize}Y>{\hsize=0.8\hsize}Y>{\hsize=0.75\hsize}Y>{\hsize=0.8\hsize}Y>{\hsize=0.8\hsize}Y}}
\toprule%
Quantum Bianchi type & $\muu_{12}^{1}$ & $\muu_{12}^{2}$ &
$\muu_{12}^{3}$ & $\muu_{23}^{1}$ & $\muu_{23}^{2}$ &
$\muu_{23}^{3}$ & $\muu_{31}^{1}$ & $\muu_{31}^{2}$ &
$\muu_{31}^{3}$
\\\otoprule
I$^{\hbar}$ & 0 & 0 & 0 & 0 & 0 & 0 & 0 & 0 & 0
\\\midrule
II$^{\hbar}$ & 0 & 0 & 0 & $\frac{\hat{p}+p_0}{2p_0}$ &
$\frac{\omega \q}{2p_0}$ & 0 & $\frac{\omega \q}{2p_0}$ &
$\frac{\hat{p}-p_0}{-2p_0}$ & 0
\\  \midrule
VII$^{\hbar}$ & 0 & 0 & 0 & $1$ & 0 & 0 & 0 & $1$ & 0
\\ \midrule
VI$^{\hbar}$ & 0 & 0 & 0 & $\frac{\hat{p}}{p_0}$& $\frac{\omega
\q}{p_0}$ & 0 & $\frac{\omega \q}{p_0}$ & $-\frac{\hat{p}}{p_0}$ & 0
\\ \midrule
IX$^{\hbar}$ & 0 & 0 & $1$ & $1$ & 0 & 0 & 0 & $1$ & 0
\\  \midrule
VIII$^{\hbar}$ & 0 & 0 & $-1$ & $1$ & 0 & 0 & 0 & $1$ & 0
\\ \midrule
V$^{\hbar}$ & $\frac{\A_-}{\sqrt{2p_0}}$ &
$\frac{-\A_+}{\sqrt{2p_0}}$ & 0 & 0 & 0 &
$\frac{-\A_-}{\sqrt{2p_0}}$ & 0 & 0 & $\frac{\A_+}{\sqrt{2p_0}}$
\\ \midrule
IV$^{\hbar}$ & $\frac{\A_-}{\sqrt{2p_0}}$ &
$\frac{-\A_+}{\sqrt{2p_0}}$ & $1$ & 0 & 0 &
$\frac{-\A_-}{\sqrt{2p_0}}$ & 0 & 0 & $\frac{\A_+}{\sqrt{2p_0}}$
\\ \midrule
VII$^{\hbar}_a$ & $\frac{a\A_-}{\sqrt{2p_0}}$ &
$\frac{-a\A_+}{\sqrt{2p_0}}$ & $1$ & $\frac{\hat{p}-p_0}{-2p_0}$ &
$\frac{\omega \q}{-2p_0}$ & $\frac{-a\A_-}{\sqrt{2p_0}}$ &
$\frac{\omega \q}{-2p_0}$ & $\frac{\hat{p}+p_0}{2p_0}$ &
$\frac{a\A_+}{\sqrt{2p_0}}$
\\  \midrule
III$_{a=1}^{\hbar}$ & $\frac{\A_-}{\sqrt{2p_0}}$ &
$\frac{-\A_+}{\sqrt{2p_0}}$ & $-1$ & $\frac{\hat{p}-p_0}{-2p_0}$ &
$\frac{\omega \q}{-2p_0}$ & $\frac{-\A_-}{\sqrt{2p_0}}$ &
$\frac{\omega \q}{-2p_0}$ & $\frac{\hat{p}+p_0}{2p_0}$ &
$\frac{\A_+}{\sqrt{2p_0}}$
\\  \midrule
VI$_{a\neq1}^{\hbar}$ & $\frac{a\A_-}{\sqrt{2p_0}}$ &
$\frac{-a\A_+}{\sqrt{2p_0}}$ & $-1$ & $\frac{\hat{p}-p_0}{-2p_0}$ &
$\frac{\omega \q}{-2p_0}$ & $\frac{-a\A_-}{\sqrt{2p_0}}$ &
$\frac{\omega \q}{-2p_0}$ & $\frac{\hat{p}+p_0}{2p_0}$ &
$\frac{a\A_+}{\sqrt{2p_0}}$
\\ \bottomrule
\end{tabularx}
\caption{Quantum counterparts of 3D real Lie algebras over the harmonic oscillator}
\label{table:Bianchi4}
\end{table}

One can easily check that 
I$^{\hbar}$, II$^{\hbar}$, VII$^{\hbar}$, VI$^{\hbar}$, IX$^{\hbar}$,
VIII$^{\hbar}$ are Lie agebras. Thus, in what follows, we will only focus on the algebras V$^{\hbar}$, IV$^{\hbar}$,
VII$_{a}^{\hbar}$, III$_{a=1}^{\hbar}$, VI$_{a\neq1}^{\hbar}$, and present the latter more compactly in a separate table.

Let $\beta,\gamma,a,b$ be real-valued parameters from Table
\ref{tabel:neli}  and let  $\mathcal{A}^{\hbar}$ denote an entry
from the first column of Table \ref{table:Bianchi4}. 
algebras V$^{\hbar}$, IV$^{\hbar}$, VII$_{a}^{\hbar}$,
III$_{a=1}^{\hbar}$, VI$_{a\neq1}^{\hbar}$ from Table
\ref{table:Bianchi4} can be presented as Table \ref{table:ah}.

\begin{table}[htpb]
\newcolumntype{Y}{>{\centering\arraybackslash}X}
\renewcommand{\tabularxcolumn}[1]{>{\arraybackslash}m{#1}}
\begin{tabularx}{\textwidth}{{>{\hsize=1.4\hsize}Y>{\hsize=0.9\hsize}Y>{\hsize=0.9\hsize}Y>{\hsize=0.9\hsize}Y>{\hsize=0.9\hsize}Y}}
\toprule%
$\mathcal{A}^{\hbar}$ & $\beta$ & $\gamma$ & $a$ & $b$
\\\otoprule
V$^{\hbar}$ & $0$ & $0$ & $1$ &  $0$  \\\midrule
IV$^{\hbar}$ & $0$ & $0$ & $1$ & $1$ \\\midrule
VII$_a^{\hbar}$ & $1$ & $1$ & $a$ & $1$ \\\midrule
III$_{a=1}^{\hbar}$ & $1$ & $1$ & $1$ & $1$  \\\midrule
VI$_{a\neq 1}^{\hbar}$ & $1$ & $1$ & $a\neq 1$ & $-1$
\\\bottomrule
\end{tabularx}
\caption{ Values of $\beta,\gamma,a,b$ for 
quantum algebras $\mathcal{A}^{\hbar}$. 
Here $a>0$}
\label{tabel:neli}
\end{table}

\begin{table}[htpb]
\newcolumntype{Y}{>{\centering\arraybackslash}X}
\renewcommand{\tabularxcolumn}[1]{>{\arraybackslash}m{#1}}
\begin{tabularx}{\textwidth}{{>{\hsize=1.55\hsize}Y>{\hsize=0.8\hsize}Y>{\hsize=0.92\hsize}Y>{\hsize=0.55\hsize}Y>{\hsize=1.45\hsize}Y
>{\hsize=0.85\hsize}Y>{\hsize=0.95\hsize}Y>{\hsize=0.85\hsize}Y>{\hsize=1.18\hsize}Y>{\hsize=0.8\hsize}Y}}
\toprule%
Quantum Bianchi type & $\muu_{12}^{1}$ & $\muu_{12}^{2}$ &
$\muu_{12}^{3}$ & $\muu_{23}^{1}$ & $\muu_{23}^{2}$ &
$\muu_{23}^{3}$ & $\muu_{31}^{1}$ & $\muu_{31}^{2}$ &
$\muu_{31}^{3}$
\\\otoprule
$\mathcal{A}^{\hbar}$ & $\frac{a\Q}{\sqrt{2p_0}}$ &
$\frac{-a\PP}{\sqrt{2\ppp}}$ & $b$ &
$\frac{-\gamma(\pp-\ppp)}{2\ppp}$ & $\frac{-\beta\omega \q}{2\ppp}$
& $\frac{-a\Q}{\sqrt{2\ppp}}$ & $\frac{-\beta\omega \q}{2\ppp}$ &
$\frac{\gamma(\pp+\ppp)}{2\ppp}$ & $\frac{a\PP}{\sqrt{2\ppp}}$
\\ \bottomrule
\end{tabularx}
\caption{$\mathcal{A}^{\hbar}$} 
\label{table:ah}
\end{table}

Let $\mathcal{A}_{HO}$ denote the state space of the quantum
harmonic oscillator and $\{e_1,e_2,\ldots\}$ be its basis. By using
Table \ref{table:ah} we define the structure equations in
$\mathcal{A}_{HO}$ by
\[
[e_i,e_j]_\hbar:=\muu_{ij}^{s} e_s
\]
where the structure operators $\muu_{ij}^{s}$ for $i,j,s\leq3$ are defined by Table \ref{table:ah} and $\muu_{ij}^{s}:=0$ for $i,j,s>3$.
For $x,y\in \mathcal{A}_{HO}$, their quantum multiplication is
defined by
\[
[x,y]_\hbar
:=\hat{\mu}^{i}_{jk} x^{j}y^{k} e_i
=\hat{\mu}^{1}_{jk} x^{j}y^{k} e_1
+\hat{\mu}^{2}_{jk} x^{j}y^{k} e_2
+\hat{\mu}^{3}_{jk} x^{j}y^{k} e_3
\]
where we omitted the trivial terms, because
$\hat{\mu}^{i}_{jk}=0$ for $i>3$.

\section{Jacobi operators}
\label{app:jacobi_calc}

For $x,y,z\in \mathcal{A}_{HO}$, their quantum Jacobi operator is
defined by
\begin{align*}
\hat{J}_\hbar(x;y;z)
&:=[x,[y,z]_\hbar]_\hbar+[y,[z,x]_\hbar]_\hbar+[z,[x,y]_\hbar]_\hbar\\
&\,\,=\hat{J}^1_\hbar(x;y;z)e_1+\hat{J}^2_\hbar(x;y;z)e_2+\hat{J}^3_\hbar(x;y;z)e_3
\end{align*}
where we again omitted the trivial terms, because
$\hat{J}^i_\hbar=0$ for $i>3$. In \cite{PV09-1} the quantum Jacobi
operators were calculated for all real three-dimensional Lie
algebras.

Denote
\begin{equation*}
(x,y,z):=\begin{vmatrix}
 x^{1} & x^{2} & x^{3} \\
 y^{1} & y^{2} & y^{3} \\
z^{1} & z^{2} & z^{3} \\
\end{vmatrix},\quad
\xi_{\pm}:=\beta\omega\hat{q}\hat{A}_{\mp} \pm \gamma(\hat{p} \mp
p_0)\hat{A}_{\pm}
\end{equation*}

\begin{thm}
The Jacobi operator components of $\mathcal{A}^{\hbar}$ read
\begin{align*}
      \hat{J}^{1}_\hbar(x;y;z)&=-\frac{a(x,y,z)}{\sqrt{2\ppp^{3}}}\qxi_+,
      \quad
      \hat{J}^{2}_\hbar(x;y;z)&=-\frac{a(x,y,z)}{\sqrt{2\ppp^{3}}}\qxi_-,
      \quad
      \hat{J}^{3}_\hbar(x;y;z)&=\frac{a^{2}(x,y,z)}{\ppp}[\PP,\Q]
 \end{align*}
\end{thm}

\begin{proof}
As an example calculate $\hat{J}^1_\hbar(x;y;z)$.
First find the products $[x,y]_\hbar,[y,z]_\hbar,[z,x]_\hbar$ in
$\mathcal{A}^{\hbar}$. Denote $\De:=(x,y,z)$ and let $\De^{ij}$ be the cofactor (signed minor) of the element of $\De$ at $i$-th row and $j$-th column.
Calculate
\begin{align*}
[x,y]_\hbar
&=[x,y]_\hbar^{i}e_i=\muu_{jk}^{i}x^{j}y^{k}e_i\\
&=\left(\muu_{12}^{1}\De^{33} -\muu_{13}^{1}\De^{32} +\muu_{23}^{1}\De^{31} \right)e_1\\
&+\left(\muu_{12}^{2}\De^{33} -\muu_{13}^{2}\De^{32} +\muu_{23}^{2}\De^{31} \right)e_2\\
&+\left(b \De^{33} -\muu_{13}^{3}\De^{32} +\muu_{23}^{3}\De^{31} \right)e_3
\\
&=\left(\frac{a\Q}{\sqrt{2p_0}}\De^{33} -\beta\frac{\omega \q}{2p_0}\De^{32} -\gamma\frac{\hat{p}-p_0}{2p_0}\De^{31} \right)e_1\\
&+\left(\frac{-a\PP}{\sqrt{2p_0}}\De^{33} +\gamma\frac{\hat{p}+p_0}{2p_0}\De^{32} -\beta\frac{\omega \q}{2p_0}\De^{31} \right)e_2\\
&+\left(b \De^{33} +\frac{a\PP}{\sqrt{2p_0}}\De^{32} -\frac{a\Q}{\sqrt{2p_0}}\De^{31} \right)e_3
\end{align*}
In the same way, we can see that
\begin{align*}
[y,z]_\hbar
&=[y,z]_\hbar^{i}e_i=\muu_{jk}^{i}y^{j}z^{k}e_i\\
&=\left(\frac{a\Q}{\sqrt{2p_0}}\De^{13} -\beta\frac{\omega \q}{2p_0}
\De^{12} -\gamma\frac{\hat{p}-p_0}{2p_0}\De^{11} \right)e_1\\
&+\left(\frac{-a\PP}{\sqrt{2p_0}}\De^{13} +\gamma\frac{\hat{p}+p_0}{2p_0}\De^{12} -\beta\frac{\omega \q}{2p_0}\De^{11} \right)e_2\\
&+\left(b \De^{13} +\frac{a\PP}{\sqrt{2p_0}}\De^{12} -\frac{a\Q}{\sqrt{2p_0}}\De^{11} \right)e_3
\end{align*}
and also
\begin{align*}
[z,x]_\hbar
&=[z,x]_\hbar^{i}e_i
=\muu_{jk}^{i}z^{j}x^{j}e_i\\
&\left(\frac{a\Q}{\sqrt{2p_0}}\De^{23}
-\beta\frac{\omega \q}{2p_0}\De^{22}
-\gamma\frac{\hat{p}-p_0}{2p_0}\De^{21} \right)e_1\\
&+
\left(\frac{-a\PP}{\sqrt{2p_0}}\De^{23}
+\gamma\frac{\hat{p}+p_0}{2p_0}\De^{22}
-\beta\frac{\omega \q}{2p_0}\De^{21} \right)e_2\\
&+
\left(b \De^{23}
+\frac{a\PP}{\sqrt{2p_0}}\De^{22}
-\frac{a\Q}{\sqrt{2p_0}}\De^{21} \right)e_3
\end{align*}
Now calculate the first component of the Jacobi operator:
\begin{align*}
\hat{J}^1_\hbar(x;y;z)
&=[x,[y,z]_\hbar]^{1}_\hbar+[y,[z,x]_\hbar]^{1}_\hbar+[z,[x,y]_\hbar]^{1}_\hbar
\\
&=\muu^{1}_{jk}x^{j}[y,z]^{k}_\hbar+\muu^{1}_{jk}y^{j}[z,x]^{k}_\hbar+\muu^{1}_{jk}z^{j}[x,y]^{k}_\hbar\\
&=\muu^{1}_{12}\left(x^{1}[y,z]^{2}_\hbar-x^{2}[y,z]^{1}_\hbar\right)+\muu^{1}_{13}\left(x^{1}[y,z]^{3}_\hbar-x^{3}[y,z]^{1}_\hbar\right)\\&+\muu^{1}_{23}\left(x^{2}[y,z]^{3}_\hbar-x^{3}[y,z]^{2}_\hbar\right)
+\muu^{1}_{12}\left(y^{1}[z,x]^{2}_\hbar-y^{2}[z,x]^{1}_\hbar\right)\\&+\muu^{1}_{13}\left(y^{1}[z,x]^{3}_\hbar-y^{3}[z,x]^{1}_\hbar\right)+\muu^{1}_{23}\left(y^{2}[z,x]^{3}_\hbar-y^{3}[z,x]^{2}_\hbar\right)\\
&+\muu^{1}_{12}\left(z^{1}[x,y]^{2}_\hbar-z^{2}[x,y]^{1}_\hbar\right)+\muu^{1}_{13}\left(z^{1}[x,y]^{3}_\hbar-z^{3}[x,y]^{1}_\hbar\right)\\&+\muu^{1}_{23}\left(z^{2}[x,y]^{3}_\hbar-z^{3}[x,y]^{2}_\hbar\right)\\
&=\frac{a\Q}{\sqrt{2p_0}} \Bigg\{ x^{1} \left(
\frac{-a\PP}{\sqrt{2p_0}}\De^{13}
+\gamma\frac{\hat{p}+p_0}{2p_0}\De^{12} -\beta\frac{\omega
\q}{2p_0}\De^{11} \right)
\\
&\qquad\qquad\quad -
x^{2}
\left(
\frac{a\Q}{\sqrt{2p_0}}\De^{13}
-\beta\frac{\omega \q}{2p_0}\De^{12}
-\gamma\frac{\hat{p}-p_0}{2p_0}\De^{11}
\right)
\Bigg\}
\\
&+
\beta\frac{\omega \q}{2p_0}
\Bigg\{
x^{1}
\left(
b \De^{13}
+\frac{a\PP}{\sqrt{2p_0}}\De^{12}
-\frac{a\Q}{\sqrt{2p_0}}\De^{11}
\right)
\\
&\qquad\qquad\quad -
x^{3}
\left(
\frac{a\Q}{\sqrt{2p_0}}\De^{13}
-\beta\frac{\omega \q}{2p_0}\De^{12}
-\gamma\frac{\hat{p}-p_0}{2p_0}\De^{11}
\right)
\Bigg\}
\\
&-
\gamma\frac{\hat{p}-p_0}{2p_0}
\Bigg\{
x^{2}
\left(
b \De^{13}
+\frac{a\PP}{\sqrt{2p_0}}\De^{12}
-\frac{a\Q}{\sqrt{2p_0}}\De^{11}
\right)
\\
&\qquad\qquad\quad -
x^{3}
\left(
\frac{-a\PP}{\sqrt{2p_0}}\De^{13}
+\gamma\frac{\hat{p}+p_0}{2p_0}\De^{12}
-\beta\frac{\omega \q}{2p_0}\De^{11}
\right)
\Bigg\}
\\
&+
\frac{a\Q}{\sqrt{2p_0}}
\Bigg\{
y^{1}
\left(
\frac{-a\PP}{\sqrt{2p_0}}\De^{23}
+\gamma\frac{\hat{p}+p_0}{2p_0}\De^{22}
-\beta\frac{\omega \q}{2p_0}\De^{21}
\right)
\\
&\qquad\qquad\quad -
y^{2}
\left(
\frac{a\Q}{\sqrt{2p_0}}\De^{23}
-\beta\frac{\omega \q}{2p_0}\De^{22}
-\gamma\frac{\hat{p}-p_0}{2p_0}\De^{21}
\right)
\Bigg\}
\\
&+
\beta\frac{\omega \q}{2p_0}
\Bigg\{
y^{1}
\left(
b \De^{23}
+\frac{a\PP}{\sqrt{2p_0}}\De^{22}
-\frac{a\Q}{\sqrt{2p_0}}\De^{21}
\right)
\\
&\qquad\qquad\quad -
y^{3}
\left(
\frac{a\Q}{\sqrt{2p_0}}\De^{23}
-\beta\frac{\omega \q}{2p_0}\De^{22}
-\gamma\frac{\hat{p}-p_0}{2p_0}\De^{21}
\right)
\Bigg\}
\\
&-
\gamma\frac{\hat{p}-p_0}{2p_0}
\Bigg\{
y^{2}
\left(
b \De^{23}
+\frac{a\PP}{\sqrt{2p_0}}\De^{22}
-\frac{a\Q}{\sqrt{2p_0}}\De^{21}
\right)
\\
&\qquad\qquad\quad -
y^{3}
\left(
\frac{-a\PP}{\sqrt{2p_0}}\De^{23}
+\gamma\frac{\hat{p}+p_0}{2p_0}\De^{22}
-\beta\frac{\omega \q}{2p_0}\De^{21}
\right)
\Bigg\}
\\
&+
\frac{a\Q}{\sqrt{2p_0}}
\Bigg\{
z^{1}
\left(
\frac{-a\PP}{\sqrt{2p_0}}\De^{33}
+\gamma\frac{\hat{p}+p_0}{2p_0}\De^{32}
-\beta\frac{\omega \q}{2p_0}\De^{31}
\right)
\\
&\qquad\qquad\quad -
z^{2}
\left(
\frac{a\Q}{\sqrt{2p_0}}\De^{33}
-\beta\frac{\omega \q}{2p_0}\De^{32}
-\gamma\frac{\hat{p}-p_0}{2p_0}\De^{31}
\right)
\Bigg\}
\\
&+
\beta\frac{\omega \q}{2p_0}
\Bigg\{z^{1}
\left(
b \De^{33}
+\frac{a\PP}{\sqrt{2p_0}}\De^{32}
-\frac{a\Q}{\sqrt{2p_0}}\De^{31}
\right)
\\
&\qquad\qquad\quad -
z^{3}
\left(
\frac{a\Q}{\sqrt{2p_0}}\De^{33}
-\beta\frac{\omega \q}{2p_0}\De^{32}
-\gamma\frac{\hat{p}-p_0}{2p_0}\De^{31}
\right)
\Bigg\}
\\
&-
\gamma\frac{\hat{p}-p_0}{2p_0}
\Bigg\{
z^{2}
\left(
b \De^{33}
+\frac{a\PP}{\sqrt{2p_0}}\De^{32}
-\frac{a\Q}{\sqrt{2p_0}}\De^{31}
\right)
\\
&\qquad\qquad\quad -
z^{3}
\left(
\frac{-a\PP}{\sqrt{2p_0}}\De^{33}
+\gamma\frac{\hat{p}+p_0}{2p_0}\De^{32}
-\beta\frac{\omega\q}{2p_0}\De^{31}
\right)
\Bigg\}
\end{align*}

Now open the parentheses and rearrange the terms. Then we have
\begin{align*}
\hat{J}^1_\hbar(x;y;z)=&-\frac{a\Q}{\sqrt{2p_0}}
\frac{a\PP}{\sqrt{2p_0}} \underbrace{ \left( x^{1}\De^{13} +
y^{1}\De^{23} + z^{1}\De^{33}
\right)}_0\\
&+
\frac{a\Q}{\sqrt{2p_0}} \gamma\frac{\hat{p}+p_0}{2p_0}
\underbrace{
\left(
x^{1}\De^{12} + y^{1}\De^{22} +z^{1}\De^{32}%
\right)}_0\\
&-
\frac{a\Q}{\sqrt{2p_0}} \beta\frac{\omega \q}{2p_0}
\underbrace{
\left(
x^{1}\De^{11} + y^{1}\De^{21} + z^{1}\De^{31}%
\right)}_{(x,y,z)} \\
&-
\frac{a\Q}{\sqrt{2p_0}} \frac{a\Q}{\sqrt{2p_0}}
\underbrace{
\left(
x^{2}\De^{13} + y^{2}\De^{23} + z^{2}\De^{33}%
\right)}_0 \\
&+
\frac{a\Q}{\sqrt{2p_0}}  \beta\frac{\omega \q}{2p_0}
\underbrace{
\left(
x^{2}\De^{12} + y^{2}\De^{22} + z^{2}\De^{32}%
\right)}_{(x,y,z)} \\
&+
\frac{a\Q}{\sqrt{2p_0}} \gamma\frac{\hat{p}-p_0}{2p_0}
\underbrace{
\left(
x^{2}\De^{11} + y^{2}\De^{21} +z^{2}\De^{31}%
\right)}_0\\
&+
 \beta\frac{\omega \q}{2p_0} b
\underbrace{
\left(
x^{1}\De^{13} + y^{1}\De^{23} +z^{1}\De^{33}%
\right)}_0  \\
&+
 \beta\frac{\omega \q}{2p_0} \frac{a\PP}{\sqrt{2p_0}}
\underbrace{
\left(
x^{1}\De^{12} + y^{1}\De^{22} +z^{1}\De^{32}%
\right)}_0 \\
 &-
 \beta\frac{\omega \q}{2p_0} \frac{a\Q}{\sqrt{2p_0}}
\underbrace{
\left(
x^{1}\De^{11} + y^{1}\De^{21} +z^{1}\De^{31}%
\right)}_{(x,y,z)}  \\
&-
 \beta\frac{\omega \q}{2p_0} \frac{a\Q}{\sqrt{2p_0}}
\underbrace{
\left(
x^{3}\De^{13} + y^{3}\De^{23} +z^{3}\De^{33}%
\right)}_{(x,y,z)}\\
&+
 \beta\frac{\omega \q}{2p_0} \beta\frac{\omega \q}{2p_0}
\underbrace{
\left(
x^{3}\De^{12} + y^{3}\De^{22} +z^{3}\De^{32}%
\right)}_0  \\
&+
 \beta\frac{\omega \q}{2p_0} \gamma\frac{\hat{p}-p_0}{2p_0}
\underbrace{
\left(
x^{3}\De^{11} + y^{3}\De^{21} +z^{3}\De^{31}%
\right)}_0  \\
&-
\gamma\frac{\hat{p}-p_0}{2p_0} b
\underbrace{
\left(
x^{2}\De^{13} + y^{2}\De^{23} +z^{2}\De^{33}
\right)}_0\\
&-
\gamma\frac{\hat{p}-p_0}{2p_0} \frac{a\PP}{\sqrt{2p_0}}
\underbrace{
\left(
x^{2}\De^{12} + y^{2}\De^{22} +z^{2}\De^{32}
\right)}_{(x,y,z)} \\
&+
\gamma\frac{\hat{p}-p_0}{2p_0} \frac{a\Q}{\sqrt{2p_0}}
\underbrace{
\left(
x^{2}\De^{11} + y^{2}\De^{21} +z^{2}\De^{31}
\right)}_0\\
&-
\gamma\frac{\hat{p}-p_0}{2p_0} \frac{a\PP}{\sqrt{2p_0}}
\underbrace{
\left(
x^{3}\De^{13} + y^{3}\De^{23} +z^{3}\De^{33}
\right)}_{(x,y,z)} \\
&+
\gamma\frac{\hat{p}-p_0}{2p_0} \gamma\frac{\hat{p}+p_0}{2p_0}
\underbrace{
\left(
x^{3}\De^{12} + y^{3}\De^{22} +z^{3}\De^{32}
\right)}_0 \\
&-
\gamma\frac{\hat{p}-p_0}{2p_0} \beta\frac{\omega \q}{2p_0}
\underbrace{
\left(
x^{3}\De^{11} + y^{3}\De^{21} +z^{3}\De^{31}
\right)}_0\\
&-\frac{a(x,y,z)}{\sqrt{2\ppp^{3}}}
(\beta\omega\hat{q}\hat{A}_{ -} + \gamma(\hat{p} - p_0)\hat{A}_{+})
\end{align*}
The remaining components $\hat{J}^2_\hbar(x;y;z)$ and
$\hat{J}^3_\hbar(x;y;z)$ can be calculated in the same way.
\end{proof}

\subsection*{Acknowledgement}

The research was in part supported by the Estonian Science Foundation, Grant ETF-6912.

\noindent
Department of Mathematics, Tallinn University of Technology\\
Ehitajate tee 5, 19086 Tallinn, Estonia\\
E-mails: eugen.paal@ttu.ee, jvirkepu@staff.ttu.ee

\begin{thebibliography}{9}
\itemsep-3pt

{\small

\bibitem{Ger}
M.~Gerstenhaber, \emph{The cohomology structure of an associative
ring}, Ann. of Math. 78 (1963), 267--288.

\bibitem{Landau80}
L. Landau and E. Lifshitz, \textit{Theoretical Physics.} Vol. 2:
\textit{Field Theory}, Nauka, Moskva, 1973 (in Russian).

\bibitem{Lax68}
P.~D.~ Lax, \emph{Integrals of nonlinear equations of evolution and
solitary waves}, Comm. Pure Applied Math. 21 (1968), 467--490.

\bibitem{Paal07}
E.~Paal, \emph{Invitation to operadic dynamics}, J. Gen. Lie Theory
Appl. 1 (2007), 57--63.

\bibitem{PV07}
E.~Paal and J.~Virkepu, \emph{Note on operadic harmonic oscillator},
Rep. Math. Phys. 61 (2008), 207--212.

\bibitem{PV08}
E.~Paal and J.~Virkepu, \emph{2D binary operadic Lax representation
for harmonic oscillator}, in: Noncommutative Structures in
Mathematics and Physics. K. Vlaam. Acad. Belgie Wet. Kunsten (KVAB),
Brussels, 2009, 209--216.

\bibitem{PV08-2}
E.~Paal and J.~Virkepu, \emph{Dynamical deformations of
three-dimensional Lie algebras in Bianchi classification over the
harmonic oscillator}, J. Math. Phys. 50 (2009), 053523.

\bibitem{PV08-1}
E.~Paal and J.~Virkepu,\emph{ Operadic representations of harmonic
oscillator in some 3d Lie algebras}, J. Gen. Lie Theory Appl. 3
(2009), 53--59.

\bibitem{PV09-1}
E.~Paal and J.~Virkepu, \emph{Quantum counterparts of
three-dimensional real Lie algebras over harmonic oscillator},
Centr. Eur. J. Phys. DOI: 102478/s11534-009-0123-8, 2009.

}

\end{thebibliography}
\end{document}